\documentclass[12pt]{amsart}

\usepackage{amsmath,amssymb,amsthm,bm,epsfig,color}

\theoremstyle{plain}
\newtheorem{theorem}{Theorem}[section]

\newtheorem{proposition}[theorem]{Proposition}

\newtheorem{lemma}[theorem]{Lemma}

\newtheorem{remark}[theorem]{Remark}

\numberwithin{equation}{section}

\newcommand{\bbR}{{\mathbb R}}

\newcommand{\cL}{{\mathcal L}}
\newcommand{\cD}{{\mathcal D}}
\newcommand{\cN}{{\mathcal N}}
\newcommand{\cH}{{\mathcal H}}
\newcommand{\cE}{{\mathcal E}}

\newcommand{\cB}{{\mathcal B}}

\newcommand{\cJ}{{\mathcal J}}

\renewcommand{\Re}{{\mathrm{Re}}}

\newcommand\ulambda{{\underline{\lambda}}}
\newcommand\Xom{{\mathcal L}^2}
\newcommand\tGamma{{\tilde{\Gamma}}}

\begin{document}

\title[]{On the generalized semi-relativistic Schr\"odinger-Poisson system in ${\mathbb R}^n$}


\author{W. Abou Salem, T. Chen and V. Vougalter}

\address{Department of Mathematics and Statistics, University of Saskatchewan, Saskatoon S7N 5E6, Canada \\ E-mail: walid.abousalem@usask.ca}

\address{Department of Mathematics, University of
Texas at Austin, Austin, TX, 78712, USA \\ E-mail: tc@math.utexas.edu}
 
\address{University of Cape Town, Department of Mathematics
and Applied Mathematics, Private Bag, Rondebosch 7701, South Africa \\E-mail: Vitali.Vougalter@uct.ac.za}

\maketitle


\begin{abstract} The Cauchy problem for the semi-relativistic Schr\"odinger-Poisson system of equations is studied in ${\mathbb R}^n, \ \ n\ge 1,$ for a wide class of nonlocal interactions.
 Furthermore, the asymptotic behavior of the solution as the mass tends to infinity is rigorously discussed, 
which corresponds to a non-relativistic limit. 
\end{abstract}


$\;$ \\



$\;$ 


\section{Introduction}
\subsection{Motivation and heuristic discussion}
In this article, we study the global Cauchy problem for the semi-relativistic Schr\"odinger-Poisson system in ${\mathbb R}^n, \ \ n\ge 1,$ for a wide class of nonlocal interactions, both in the attractive and repulsive cases. This system is relevant to the description of many-body semi-relativistic quantum particles in the mean-field limit. 
We consider a system of $N$ semi-relativistic quantum particles in $ {\mathbb R}^n,\ \ n\ge 1$ with long-range two-body interactions $g\frac{1}{N}\sum_{1\le i<j\le N}\frac{1}{|x_i-x_j|^\gamma}$,  
with $0<\gamma\le 1$ if $n\geq2$, and $0<\gamma<1$ if $n=1$,
and with $g\in {\mathbb R}$.
In the mean-field limit, one can formally show that the density matrix that describes the {\it mixed} state of the system satisfies the Hartree-von Neumann equation
\begin{equation}
\label{eq:HartreeVonNeumann}
\begin{cases}
i\partial_t \rho(t) = [H_m + w_\gamma\star n(t),\rho(t)], \ \ x\in {\mathbb R}^n, \ \ n\ge 1, \ \ t\ge 0\\
H_m=\sqrt{m^2-\Delta}-m, \ \  w_\gamma = g \frac{1}{|x|^\gamma},\ \ n(t,x)=\rho(t,x,x),
\ \rho(0)=\rho_0
\end{cases},
\end{equation}
where $\Delta$ stands for the $n$-dimensional Laplacian, $\star$ stands for convolution in ${\mathbb R}^n,$ and $m\ge 0$ is the mass.\footnote{The rigorous derivation of the semi-relativistic Hartree-von Neumann equation is a topic of future work, see \cite{AS09, A10} for a derivation of this system of equations in the non-relativistic case.} 
Since $\rho_0$ is a positive, self-adjoint trace-class operator acting on $L^2({\mathbb R}^n),$ its kernel can be decomposed with respect to an  orthonormal basis of $L^2({\mathbb R}^n),$ 
\begin{equation}
\label{eq:kernelinit}
\rho_0(x,y) = \sum_{k\in {\mathbb N}} \lambda_k \psi_k(x)\overline{\psi_k(y)}
\end{equation}
where $\{\psi_k\}_{k\in {\mathbb N}}$ denotes an orthonormal basis of $L^2({\mathbb R}^n).$ 
Furthermore,
\begin{equation*}
	\ulambda:=\{\lambda_k\}_{k\in {\mathbb N}}\in l^1
	\;, \;\;
	\lambda_k\ge 0
	\; , \; \;
	\sum_k\lambda_k=1.
\end{equation*}
We will show that there exists a one-parameter family of 
complete orthonormal bases of $L^2({\mathbb R}^n)$, $\{\psi_k(t)\}_{k\in {\mathbb N}}$, for $t\in\mathbb{R}_+$,  such that  the kernel of the solution $\rho(t)$ to \eqref{eq:HartreeVonNeumann} can be represented as 
\begin{equation}
\label{eq:kernel}
\rho(t,x,y) = \sum_{k\in {\mathbb N}} \lambda_k \psi_k(t,x)\overline{\psi_k(t,y)}.
\end{equation}
Substituting (\ref{eq:kernel}) in (\ref{eq:HartreeVonNeumann}), the one-parameter family of orthonormal vectors  $\{\psi_k(t)\}_{k\in {\mathbb N}}$ is seen to satisfy the semi-relativistic Schr\"odinger-Poisson system
\begin{equation}
\label{eq:SP}
i\frac{\partial \psi_{k}}{\partial t}=H_{m}\psi_{k}+V\psi_{k}, \quad k\in 
{\mathbb N}
\end{equation}  
\begin{equation}
\label{eq:V}
V[\Psi]=w_\gamma \star n[\Psi], \quad \Psi:=\{\psi_{k}\}_{k=1}^{\infty},
\end{equation} 
\begin{equation} 
\label{eq:n}
n[\Psi(x,t)]=\sum_{k=1}^{\infty}\lambda_{k}|\psi_{k}|^{2} \,.
\end{equation} 

The purpose of this note is to show global well-posedness of (\ref{eq:SP}) in a suitable Banach space (to be specified below), and to study the asymptotics of the solution as the mass $m$ tends to $\infty,$ which corresponds to the non-relativistic limit, see \cite{S76}. The semi-relativistic Schr\"odinger-Poisson system of equations in a finite domain of ${\mathbb R}^3$ and with repulsive Coulomb interactions has been studied recently in \cite{ACV12,ACV12-2}. Here, we generalize the result of \cite{ACV12} in several ways. First, the problem is studied in ${\mathbb R}^n, \ \ n\ge 1.$ Second, we consider a wide class of nonlocal interactions in both the attractive and repulsive cases, and which includes the repulsive Coulomb case in three spatial dimensions. Third, in the non-relativistic limit $m\rightarrow\infty,$ we recover the non-relativistic Schr\"odinger-Poisson system of equations, which has been studied extensively, see for example \cite{BM91, IZ94} and references therein. In the special case when the initial density matrix is a pure state $\rho_0 = |\psi_0\rangle\langle\psi_0|,$ the Schr\"odinger-Poisson system becomes a single Hartree equation 
$$i\partial_t \psi = (\sqrt{m^2-\Delta}-m) \psi + (w_\gamma\star |\psi|^2)\psi, \ \ \psi(0)=\psi_0.$$
In that sense, our analysis generalizes the results of \cite{L07, CO07} to the effective dynamics of a {\it mixed state} of a semi-relativistic system.

The organization of this paper is as follows. In Subsection \ref{sec:Main} we state our main results. We prove local and global well-posedness in Section \ref{sec:Well-posedness}. Finally, in Section \ref{sec:Nonrelativistic}, we discuss the asymptotic behavior of the solutions as the mass tends to infinity. For the benefit of a general reader, we recall some useful results about fractional integration and fractional Leibniz rule in Appendix \ref{sec:Preliminaries}.

\subsection{Notation}

\begin{itemize}

\item $A\lesssim B$ means that there exists a positive constant $C$ independent mass $m$ such that $A\le C\; B.$ 
\item $L^p$ stands for the standard Lebesgue space. Furthermore, $L_I^pB = L^p(I;B).$ $\langle \cdot, \cdot \rangle_{L^2}$ denotes the $L^2(\bbR^n)$ inner product. We will often use the abbreviated notation $L^p_T$ for $L^p_{[0,T]}$, 
in the situation where $[0,T]$ denotes a time interval.

\item $l^1 = \{\{a_l\}_{l\in {\mathbb N}} | \ \ \sum_{l\ge 1} |a_l| <\infty\}.$
\item   $W^{s,p} = (-\Delta+1)^{-\frac{s}{2}}L^p,$ the standard (complex) Sobolev space. When $p=2$, $W^{s,2}= H^s$.
$\dot{H}^s$ denotes the homogeneous Sobolev space with norm 
$\|\psi\|_{\dot{H}^s} = (\langle \psi, (-\Delta)^s \psi\rangle_{L^2})^{\frac{1}{2}}.$
\item For fixed $\ulambda\in l^1, \ \ \lambda_k\ge 0,$ and for sequences of functions
$\Phi:=\{\phi_{k}\}_{k\in {\mathbb N}}$ and 
$\Psi:=\{\psi_{k}\}_{k\in {\mathbb N}},$ we define the inner product
$$
\langle\Phi,\Psi\rangle_{\Xom}:=\sum_{k\ge 1}\lambda_{k}\langle\phi_{k}, 
\psi_{k}\rangle_{L^{2}},
$$ 
which induces the norm 
$$
\|\Phi\|_{\Xom}=(\sum_{k\ge 1}\lambda_{k}{\|\phi_{k}\|
_{L^{2}}^{2}})^{\frac{1}{2}}.
$$
The corresponding Hilbert space is $\Xom.$
\item For fixed $\ulambda\in l^1, \ \ \lambda_k\ge 0,$ $$\cH^s = \{ \Psi=\{\psi_k\}_{k\in {\mathbb N}} | \ \ \psi_k\in H^s, \ \ \sum_{k\ge 1} \lambda_k \|\psi_k\|_{H^s}^2<\infty \}$$ is a Banach space with norm
$\|\Psi\|_{\cH^s} = (\sum_{k\ge 1}\lambda_k \|\psi_k\|_{H^s}^2)^{\frac{1}{2}}.$

\item For fixed $\ulambda\in l^1, \ \ \lambda_k\ge 0,$ $$\dot{\cH}^s = \{ \Psi=\{\psi_k\}_{k\in {\mathbb N}} | \ \ \psi_k\in \dot{H}^s, \ \ \sum_{k\ge 1} \lambda_k \|\psi_k\|_{\dot{H}^s}^2<\infty\}$$ is a Banach space with norm
$\|\Psi\|_{\dot{\cH}^s} = (\sum_{k\ge 1}\lambda_{k}\|\psi_{k}\|_{\dot{H}^s}^{2})^{\frac{1}{2}}.$

\end{itemize}

\subsection{Statement of main results}\label{sec:Main}

For $s\ge 1/2,$ we define the state space for the Schr\"odinger-Poisson system  by
$$
{\mathcal S}^s:=\{ (\Psi, \ulambda)  | \ \  \Psi=\{\psi_{k}\}_{k=1}\in \cH^s \; \; is \; a \; complete \; 
orthonormal \; system \; in \; L^{2}({\mathbb R}^n), 
$$
$$
\ulambda = \{\lambda_k\}_{k\in{\mathbb N}}\in l^1, \; \lambda_k\ge 0\}.
$$

The following is our first main result about the global Cauchy problem.

\begin{theorem}\label{th:Well-posedness}
Consider the system of equations (\ref{eq:SP})-(\ref{eq:n}), with $m\ge 0$,
with $0<\gamma\le 1$ if $ n\ge 2,$ and $0<\gamma<1$ if $n=1$, and let $s\ge \gamma/2$. 
Suppose that $(\Psi(0),\ulambda)\in {\mathcal S}^s.$ If $g\ge 0,$ or $g<0$ with $\|\Psi(0)\|_{\Xom}$ small enough, then
there is a unique mild solution $(\Psi,\ulambda)\in C([0,\infty], {\mathcal S}^s).$  
\end{theorem}

\begin{remark}
$\underline{\lambda}$ is time-independent, and hence the evolution can be thought as that of 
$\Psi\in \cH^s.$
\end{remark}

\begin{remark}\label{rm:uniformbd}
It follows from the proof of local well-posedness (Proposition \ref{pr:local} in Section \ref{sec:local}) that there exists a positive time $T$ independent of $m\ge 0$ such that $\|\Psi\|_{L^\infty_T \cH^s} \le C\|\Psi(0)\|_{\cH^s},$ 
where $C>0$ is independent of $m.$
\end{remark}

\begin{remark}\label{rm:zeromass}
The solution is continuous in the mass $m.$ In particular, as $m\searrow 0,$ and for $T>0$ fixed, $\Psi\rightarrow\Psi^{(0)}$ strongly in $L^\infty_T(\cH^s),$ where $\Psi^{(0)}$ satisfies (\ref{eq:SP})-(\ref{eq:n}) with initial condition $\Psi(0),$ see Proposition \ref{pr:zeromass} in Sect. \ref{sec:Well-posedness}. 
\end{remark}

The second result is about the infinite mass limit. 
Let $\Gamma$ satisfy the nonrelativistic Schr\"odinger-Poisson system of equations
\begin{equation*}
i\frac{\partial \psi_{k}}{\partial t}=-\frac{1}{2m}\Delta\psi_{k}+V\psi_{k}, \quad k\in 
{\mathbb N}
\end{equation*}  
\begin{equation*}
V[\Psi]=w_\gamma \star n[\Psi], \quad \Psi:=\{\psi_{k}\}_{k=1}^{\infty},
\end{equation*} 
\begin{equation*} 
n[\Psi(x,t)]=\sum_{k=1}^{\infty}\lambda_{k}|\psi_{k}|^{2},
\end{equation*} 
with initial condition $\Psi(0)=\{\psi_k(0)\}_{k\in {\mathbb N}}$.

\begin{theorem}\label{th:InfiniteMass}
Suppose that the hypotheses of Theorem \ref{th:Well-posedness} hold. Then there exists $\tau>0$ such that $\Psi\rightarrow \Gamma$ in $L^\infty_\tau(\cH^s)$ as $m\rightarrow\infty.$ 
\end{theorem}

In other words, when the mass tends to infinity, the solution of the semi-relativistic 
Schr\"odinger-Poisson system of equations behaves like the nonrelativistic one.

\section{Well-posedness}\label{sec:Well-posedness}

\subsection{Local well-posedness}\label{sec:local}

In what follows, we fix $\underline{\lambda}\in l^1, \ \ \lambda_l\ge 0, \ \ l\in {\mathbb N}.$
We start by showing that the nonlinearity $V[\Psi]\Psi$ is locally Lipschitz.

\begin{lemma}\label{lm:Lipschitz}
For $\Psi,\Phi\in \cH^s,$
$$\|V[\Psi]\Psi - V[\Phi]\Phi\|_{\cH^s} \lesssim (\|\Psi\|_{\cH^s}^2+\|\Phi\|_{\cH^s}^2) \|\Psi-\Phi\|_{\cH^s}.$$
\end{lemma}

\begin{proof}
The proof relies on the fractional Leibniz rule and fractional integration, see Appendix \ref{sec:Preliminaries}. From the Minkowski inequality, 
\begin{equation}
\|V[\Psi]\Psi - V[\Phi]\Phi\|_{\cH^s} \lesssim \|(V[\Psi] - V[\Phi])\Psi\|_{\cH^s} + \|V[\Phi](\Psi-\Phi)\|_{\cH^s} \label{eq:VDiff1}
\end{equation}
We begin by estimating the first term on the right.
\begin{align}
&\|(V[\Psi] - V[\Phi])\Psi\|_{\cH^s} \lesssim \sum_{k,l\ge 1} \lambda_k\lambda_l \|w_\gamma\star (|\psi_l|^2 -|\phi_l|^2)\psi_k\|_{H^s} \nonumber\\
&\lesssim \sum_{k,l\ge 1}\lambda_k\lambda_l\{ \|w_\gamma\star (|\psi_l|^2 -|\phi_l|^2)\|_{L^\infty}\|\psi_k\|_{H^s} + \|w_\gamma\star (|\psi_l|^2 -|\phi_l|^2)\|_{W^{s,\frac{2n}{\gamma}}}\|\psi_k\|_{L^\frac{2n}{n-\gamma}} \}\nonumber\\
&\lesssim \sum_{k,l\ge 1}\lambda_k\lambda_l\{ \|\psi_l -\phi_l\|_{H^{\frac{\gamma}{2}}}(\|\psi_l\|_{H^{\frac{\gamma}{2}}}+\|\psi_l\|_{H^{\frac{\gamma}{2}}})\|\psi_k\|_{H^s} + \||\psi_l|^2 -|\phi_l|^2\|_{L^{\frac{2n}{2n-\gamma}}}\|\psi_k\|_{H^{\frac{\gamma}{2}}} \}\nonumber\\
&\lesssim (\|\Psi\|_{\cH^s}^2+\|\Phi\|_{\cH^s}^2) \|\Psi-\Phi\|_{\cH^s}.\label{eq:VDiff2}
\end{align}
Here, we used Minkowski inequality in the first line, fractional Leibniz rule (Lemma \ref{lm:fLeibniz} in the Appendix) in the second line, H\"older's inequality, fractional integration (Lemma \ref{lm:fIntegralOperator}) and Lemma \ref{lm:fSmoothing} in the third line.
Similarly, 
\begin{align}
& \|V[\Phi](\Psi-\Phi)\|_{\cH^s} \lesssim \sum_{k,l\ge 1} \lambda_k\lambda_l \|w_\gamma\star |\phi_l|^2(\psi_k-\phi_k)\|_{H^s} \nonumber \\
&\lesssim \sum_{k,l\ge 1}\lambda_k\lambda_l\{ \|w_\gamma\star |\phi_l|^2\|_{L^\infty}\|\psi_k-\phi_k\|_{H^s} + 
\|w_\gamma\star |\phi_l|^2\|_{W^{s,\frac{2n}{\gamma}}}\|\psi_k-\phi_k\|_{L^\frac{2n}{n-\gamma}} \}\nonumber\\
&\lesssim \sum_{k,l\ge 1}\lambda_k\lambda_l\{ \|\phi_l\|_{H^{\frac{\gamma}{2}}}^2\|\psi_k-\phi_k\|_{H^s} + \||\phi_l|^2\|_{L^{\frac{2n}{2n-\gamma}}}\|\psi_k-\phi_k\|_{H^{\frac{\gamma}{2}}} \}\nonumber\\
&\lesssim \|\Phi\|_{\cH^s}^2 \|\Psi-\Phi\|_{\cH^s} \label{eq:VDiff3}.
\end{align}
The claim of the lemma follows from inequalities (\ref{eq:VDiff1}), (\ref{eq:VDiff2}) and (\ref{eq:VDiff3}).
\end{proof}

Using a standard contraction map argument, the generalized semi-relativistic Schr\"odinger-Poisson system of equations is locally well-posed. 

\begin{proposition}\label{pr:local}
Consider the system of equations (\ref{eq:SP})-(\ref{eq:n}), with $m\ge 0$, $0<\gamma\le 1$ if $n\ge 2$, 
and $0<\gamma<1$ if $n=1$. Suppose that $(\underline{\lambda},\Psi(0))\in {\mathcal L}^s,\; s\ge \gamma/2.$ Then there exists a positive time $T$ such that the unique solution $\Psi \in C([0,T]; {\cH}^s).$ Furthermore, there exists a maximal time $\tau^*\in (0,\infty]$ such that $\lim_{t\nearrow\tau^*}\|\Psi(t)\|_{\cH^s}=\infty.$ 
\end{proposition}

\begin{proof}
Given $\rho,T>0,$ consider the Banach space $${\mathcal B}^s_{T,\rho}=\{\Psi \in L^\infty_T (\cH^s): \; \|\Psi\|_{L^\infty_T\cH^s}\le\rho\}.$$ 
Let $U^{(m)} = e^{-itH_m},$  the unitary operator generated by the semi-relativistic Hamiltonian $H_m= \sqrt{-\Delta +m^2}-m. $
We define the mapping $\cN$  by
$${\cN}(\Psi)(t) = U^{(m)}(t) \Psi(0) - i\int_0^t U^{(m)}(t-t') V[\Psi(t')] \Psi(t')dt',$$
which is the solution given by the Duhamel formula. 
First we show that $\cN$ is a mapping from $\cB^s_{T,\rho}$ into itself.
\begin{align*}
\|\cN (\Psi)\|_{L^\infty_T \cH^s} &\le \|\Psi(0)\|_{\cH^s} + T \|V[\Psi]\Psi\|_{L^\infty_T \cH^s}\\
&\le \|\Psi(0)\|_{\cH^s} + T \sum_{k,l\ge 1} \lambda_k \lambda_l \|w_\gamma \star |\psi_l|^2 \psi_k\|_{L^\infty_T H^s}\\
&\le  \|\Psi(0)\|_{\cH^s} + T \sum_{k,l\ge 1} \lambda_k \lambda_l \{\|w_\gamma\star(|\psi_l|^2)\|_{L^\infty_TL^\infty} \|\psi_k\|_{L^\infty_T H^s} +  \\ & +  \|w_\gamma\star(|\psi_l|^2)\|_{L^\infty_T W^{s,\frac{2n}{\gamma}}} \|\psi_k\|_{L^\infty_T L^{\frac{2n}{n-\gamma}}}\} ,
\end{align*}
where we have used fractional Leibniz rule (Lemma \ref{lm:fLeibniz}) in the last inequality. It follows from fractional integration (Lemma \ref{lm:fIntegralOperator}) and Sobolev embedding $H^{\frac{\gamma}{2}}\hookrightarrow L^{\frac{2n}{n-\gamma}}$ that 
\begin{align*}
\|\cN (\Psi)\|_{L^\infty_T \cH^s} &\le \|\Psi(0)\|_{\cH^s} + T \sum_{k,l\ge 1} \lambda_k \lambda_l \{\|\psi_l\|^2_{L^\infty_TH^{\frac{\gamma}{2}}} \|\psi_k\|_{L^\infty_T H^s} + \\ & +  \|\psi_l\|^2_{L^\infty_TL^{\frac{2n}{n-\gamma}}} \|\psi_k\|_{L^\infty_T H^s}\}\\
&\le \|\Psi(0)\|_{\cH^s} + T \sum_{k,l\ge 1} \lambda_k \lambda_l \{\|\psi_l\|^2_{L^\infty_TH^{\frac{\gamma}{2}}} \|\psi_k\|_{L^\infty_T H^s}\}\\
&\le \|\Psi(0)\|_{\cH^s} + T (\sum_{l\ge 1} \lambda_l \{\|\psi_l\|^2_{L^\infty_TH^{\frac{\gamma}{2}}} )(\sum_{k\ge 1}\lambda_k\|\psi_k\|^2_{L^\infty_T H^s})^{\frac{1}{2}}\\
&\le \|\Psi(0)\|_{\cH^s} + T \|\Psi\|_{L^\infty_T \cH^{\frac{\gamma}{2}}}^2 \|\Psi\|_{L^\infty_T \cH^s}, 
\end{align*}
where we have used the fact that $\lambda_k\ge 1 $ and $\sum_{k\ge 1}\lambda_k =1$ before the last inequality. 

Since $s\geq\frac\gamma2$, and since by assumption, $\Psi\in B_{T,\rho}^s$, we can choose $T$ and $\rho$ such that 
\begin{equation*}
\|\Psi(0)\|_{\cH^s}\le \frac{\rho}{2}, \ \ T\rho^2 <\frac{1}{2},
\end{equation*}
it follows from the last inequality and the Duhamel formula that  
$$\|\Psi\|_{L^\infty_T \cH^s} \le 2\|\Psi(0)\|_{\cH^s}\le \rho.$$ 

Second, since the nonlinearity is locally Lipschitz (Lemma \ref{lm:Lipschitz}), $\cN$ is a contraction map for sufficiently small $T.$
\begin{align*}
&\|\cN(\Psi) - \cN(\Phi)\|_{L^\infty_T\cH^s} \le  T \|V[\Psi]\Psi - V[\Phi]\Phi\|_{L^\infty_T \cH^s}\\
&\lesssim  T\rho^2 \|\Psi-\Phi\|_{L^\infty_T \cH^s}.
\end{align*}
Local well-posedness follows from a standard contraction mapping argument, see for example, \cite{Cazenave96}.
\end{proof}

It follows from local well-posedness that for every $k\in {\mathbb N},$ $\|\psi_k\|_{L^2}$ is conserved.

\begin{lemma}\label{lm:charge}
Suppose that the hypotheses of Proposition \ref{pr:local} hold. Then $\|\psi_k(t)\|_{L^2} = \|\psi_k(0)\|_{L^2}, \ \ t\in [0,\tau^*).$
\end{lemma}

\begin{proof}
Multiplying (\ref{eq:SP}) by $\overline{\psi_k}$ and integrating over space yields
$$\frac{i}{2}\partial_t \|\psi_l\|^2 = \langle \psi_l, H_m\psi_l\rangle + \langle \psi_l , V[\Psi]\psi_l\rangle.$$
Taking the imaginary part of both sides of the equation yields $\partial_t \|\psi_l\|^2=0.$
\end{proof}

The energy functional associated with the semi-relativistic Schr\"odinger-Poisson system is 
$$\cE(\Psi) = \frac{1}{2}\langle \Psi, H_m\Psi\rangle_{\cL^2} + \frac{1}{4}\langle \Psi, V[\Psi]\Psi\rangle_{\cL^2} .$$ Formally, conservation of energy follows from multiplying (\ref{eq:SP}) by $\lambda_l \partial_t \overline{\psi_k},$ integrating over space, and summing over $k\ge 1.$ To make the argument precise, we need a regularization procedure.

\begin{lemma}\label{lm:energy}
Suppose that the hypotheses of Proposition \ref{pr:local} hold. 
Then $\cE(\Psi(t)) = \cE(\Psi(0))$, $t\in [0,\tau^*)$,
is satisfied for solutions $\Psi\in C([0,\tau^*),\cH^s)$ 
with $s\geq\frac12$.

\end{lemma}

\begin{proof}
Let
$$\cJ_\epsilon = (\epsilon H_m + 1)^{-1}, \ \ \epsilon >0,$$
act on the sequence of embedding spaces 
$$\cdots \ \ \cH^{\frac{3}{2}} \hookrightarrow  \cH^{\frac{1}{2}} \hookrightarrow  \cH^{\frac{-1}{2}} \hookrightarrow  \cH^{\frac{-3}{2}} \ \ \cdots $$
It follows from fractional calculus that 
\begin{itemize}
\item[(i)] $\cJ_\epsilon$ is a bounded operator from $\cH^s$ to $\cH^{s+1},$
\item[(ii)] $\|\cJ_\epsilon \Psi\|_{\cH^s}\le \|\Psi\|_{\cH^s},$ and
\item[(iii)] $\cJ_\epsilon\Psi\rightarrow \Psi$ strongly in $\cH^s$ as $\epsilon\rightarrow 0.$
\end{itemize}
Now,
\begin{align*}
\cE(\cJ_\epsilon \Psi(t_2)) &- \cE(\cJ_\epsilon \Psi(t_1)) = \int_{t_1}^{t_2} \partial_t \cE(\cJ_\epsilon \Psi(t))~dt \\
=& \; \Re\Big\{\int_{t_1}^{t_2} -i\langle H_m \cJ_\epsilon \Psi(t), H_m \cJ_\epsilon \Psi(t)\rangle_{\cL^2} +
\\&+ \langle H_m \cJ_\epsilon \Psi(t), \cJ_\epsilon V[\Psi(t)]\Psi(t)\rangle_{\cL^2} + \\ 
& + \langle \cJ_\epsilon V[\cJ_\epsilon \Psi(t)] \cJ_\epsilon\Psi(t), H_m \cJ_\epsilon\Psi(t)\rangle_{\cL^2} +
\\&+ \langle \cJ_\epsilon V[\cJ_\epsilon \Psi(t)] \cJ_\epsilon\Psi(t), \cJ_\epsilon V[\Psi(t)]\Psi(t)\rangle_{\cL^2}\Big\}.
\end{align*}
The first term is trivially zero, since $H_m\cJ_\epsilon = \cJ_\epsilon H_m.$ Let 
\begin{align*}
g_\epsilon (t) = &\; \Re\{ \langle H_m \cJ_\epsilon \Psi(t), \cJ_\epsilon V[\Psi(t)]\Psi(t)\rangle_{\cL^2} +  
\\&+ \langle \cJ_\epsilon V[\cJ_\epsilon \Psi(t)] \cJ_\epsilon\Psi(t), H_m \cJ_\epsilon\Psi(t)\rangle_{\cL^2} + 
\\ &+ \langle \cJ_\epsilon V[\cJ_\epsilon \Psi(t)] \cJ_\epsilon\Psi(t), \cJ_\epsilon V[\Psi(t)]\Psi(t)\rangle_{\cL^2}\}.
\end{align*}
Then 
$$\cE(\cJ_\epsilon \Psi(t_2)) - \cE(\cJ_\epsilon \Psi(t_1)) = \int_{t_1}^{t_2} g_\epsilon (t)dt.$$
It follows from the above properties (i)-(iii) of $\cJ_\epsilon$ that $\lim_{\epsilon\rightarrow 0}g_\epsilon (t)=0.$ Furthermore, 
\begin{equation}
g_\epsilon(t) \le \|V[\Psi(t)] \Psi(t)\|_{\cL^2} \|H_m\Psi(t)\|_{\cL^2} + \|V[\Psi(t)] \Psi(t)\|_{\cL^2}^2. \label{eq:gBd} 
\end{equation}
Using Lemma \ref{lm:fSmoothing}, we have 
\begin{equation*}
\|V[\Psi] \Psi\|_{\cL^2} \lesssim \sum_{k,l\ge 1} \lambda_k\lambda_l \|\psi_l\|^2_{\dot{H}^{\frac{\gamma}{2}}} \|\psi_k\|_{L^2}.
\end{equation*}
The Gagliardo-Nirenberg inequality, 
$$\|\psi_l\|_{\dot{H}^{\frac{\gamma}{2}}}  \lesssim \|\psi_l\|^{\gamma}_{\dot{H}^{\frac{1}{2}}}\|\psi_l\|_{L^2}^{1-\gamma},$$ together with conservation of charge (Lemma \ref{lm:charge}), yields  
\begin{align*}
\|V[\Psi] \Psi\|_{\cL^2} &\lesssim \sum_{l\ge 1} \lambda_l \|\psi_l\|_{\dot{H}^{\frac{1}{2}}}^{2\gamma}\\
& \lesssim (\sum_{l\ge 1} \lambda_l \|\psi_l\|_{\dot{H}^{\frac{1}{2}}}^2)^{\gamma}\\
&\lesssim \|\Psi\|_{\cH^{\frac{1}{2}}}^{2\gamma} ,
\end{align*}
where we have used in the second inequality the fact that $\sum_{l\ge 1}\lambda_l=1,  \ \ \lambda_l\ge 0,$ and $f(x)=x^\gamma, \ \ 0<\gamma<1,$ is concave (equality when $\gamma=1$ is trivially satisfied).
Substituting back in (\ref{eq:gBd}) yields
$$g_\epsilon(t) \lesssim \|\Psi\|_{\cH^{\frac{1}{2}}}^{2\gamma+1} +\|\Psi\|_{\cH^{\frac{1}{2}}}^{4\gamma},$$
which is finite for $t<\tau^*.$ By the Dominated Convergence Theorem, $$\cE(\Psi(t_2)) - \cE(\Psi(t_1)) = \int_{t_1}^{t_2} \lim_{\epsilon\rightarrow0} g_\epsilon(t)dt = 0 \,,$$
as claimed.
\end{proof}

Global well-posedness follows from conservation of charge and energy.

\begin{proposition}\label{pr:global}
Suppose that the hypotheses of Proposition \ref{pr:local} hold. Then, if $g>0$ or $g<0$ with $\|\Psi(0)\|_{\cL^2}$ small enough, 
$$\|\Psi(t)\|_{\cH^s} \le C \|\Psi(0)\|_{\cH^s} e^{\alpha\left(\cE(\Psi(0)) +\|\Psi(0)\|^{\delta}_{\cL^2} \right)t},$$
where $C,\alpha$ and $\delta$ are positive constants that are independent of $m\ge 0.$
\end{proposition}

\begin{proof}
We start by bounding $\|\Psi(t)\|_{\dot{\cH}^{\frac{\gamma}{2}}}$ from above, uniformly in time. 
\begin{align*}
\langle \Psi, V[\Psi]\Psi\rangle_{\cL^2} & = \sum_{l\ge 1} \lambda_ l \langle \psi_l , V[\Psi]\psi_l\rangle \\
&\le \|V[\Psi]\|_{L^\infty} \|\Psi\|_{\cL^2}^2\\
&\lesssim (\sum_{k\ge 1} \lambda_k \|\psi_k\|_{\dot{H}^{\frac{\gamma}{2}}}^2) \|\Psi\|_{\cL^2}^2\\
&\lesssim (\sum_{k\ge 1} \lambda_k \|\psi_k\|_{\dot{H}^{\frac{1}{2}}}^{2\gamma}) \|\Psi\|_{\cL^2}^2\\
&\lesssim (\sum_{k\ge 1} \lambda_k \|\psi_k\|_{\dot{H}^{\frac{1}{2}}}^{2})^\gamma \|\Psi\|_{\cL^2}^2\\
&\lesssim \|\Psi\|_{\cH^{\frac{1}{2}}}^\gamma \|\Psi\|_{\cL^2}^2.
\end{align*}
Here, we used H\"older's inequality in the second line, Lemma \ref{lm:fSmoothing} in the third line, the Gagliardo-Nirenberg inequality and conservation of charge in the fourth line, and $\sum_{k\ge 1}\lambda_k=1, \ \ \lambda_k\ge 0,$ the fact that $x^\gamma, \ \ 0<\gamma<1,$ is concave in the fifth line (equality when $\gamma=1$ is trivially satisfied). 
Together with conservation of energy (Lemma \ref{lm:energy}), this implies that for $g>0$ or $g<0$ with $\|\Psi(0)\|_{\cL^2}$ small enough,
\begin{equation}
\|\Psi\|_{\dot{\cH}^{\frac{\gamma}{2}}} \le \alpha\left(\cE(\Psi(t)) + \|\Psi(0)\|^\delta_{\cL^2}\right),\label{eq:PsiFracNorm}
\end{equation}
where $\alpha$ and $\delta$ are constants independent of the mass $m\ge 0.$
Now, it follows from the Duhamel formula that 
\begin{align*}
\|\Psi(t)\|_{\cH^s} &\le \|\Psi(0)\|_{\cH^s} + \int_0^t \|\Psi(t')\|_{\dot{\cH}^{\frac{\gamma}{2}}}^2 \|\Psi(t')\|_{\cH^s} ~dt'\\
&\le \|\Psi(0)\|_{\cH^s} + \alpha\left(\cE(\Psi(t)) + \|\Psi(0)\|^\delta_{\cL^2}\right) \int_0^t \|\Psi(t')\|_{\cH^s} ~dt',
\end{align*}
where we used H\"older's and Minkowski inequalities in the first line, and (\ref{eq:PsiFracNorm}) in the second line.
By Gronwall's lemma, 
$$\|\Psi(t)\|_{\cH^s}\le \|\Psi(0)\|_{\cH^s} e^{\alpha\left(\cE(\Psi(t)) + \|\Psi(0)\|^\delta_{\cL^2}\right) t}$$
follows.
\end{proof}

\begin{proof}[Proof of Theorem \ref{th:Well-posedness}]
It follows from Propositions \ref{pr:local} and \ref{pr:global} that $\tau^*=\infty,$ i.e., the generalized semi-relativistic Schr\"odinger-Poisson system of equations is globally well-posed.
\end{proof}

We now prove the claim of Remark \ref{rm:zeromass} about the asymptotic behaviour of the system as the mass tends to zero.

\begin{proposition}\label{pr:zeromass}
Consider the system of equations (\ref{eq:SP})-(\ref{eq:n}) with initial condition 
$(\underline{\lambda},\Psi(0)).$ Let $\Psi^{(0)}$ denote the solution of the initial value problem with mass $m=0,$ and fix $T>0.$ Under the hypotheses of Proposition \ref{pr:global}, 
$\Psi\rightarrow \Psi^{(0)}$ strongly in $L^\infty_T(\cH^s)$ as $m\rightarrow 0.$ 
\end{proposition}

\begin{proof}
Proposition \ref{pr:global} implies that, given $T>0,$ there exists finite $\rho>0$ such that 
\begin{equation}
\sup_{m\in [0,1]}\|\Psi\|_{L^\infty_T\cH^s}<\rho. \label{eq:NormBd1}
\end{equation}
We now compare the norm of the difference of $\Psi(t)$ and $\Psi^{(0)}(t),$ $t\in [0,T].$ It follows from the Duhamel formula that 
\begin{align*}
\|\Psi(t) &- \Psi^{(0)}(t)\|_{\cH^s} \lesssim \|\left(U^{(m)}(t)-U^{(0)}(t) \right) \Psi(0)\|_{\cH^s} +
\\ & + \int_0^t \{\|V[\Psi(t')]\Psi(t') - V[\Psi^{(0)}(t')]\Psi^{(0)}(t')\|_{\cH^s}  + \\ \ \ &+ \|\left(U^{(m)}(t')-U^{(0)}(t') \right) V[\Psi^{(0)}(t')]\Psi^{(0)}(t')\|_{\cH^s}\} dt'\\
\lesssim& \; mT\|\Psi(0)\|_{\cH^s} +   \int_0^t\|V[\Psi(t')]\Psi(t') - V[\Psi^{(0)}(t')]\Psi^{(0)}(t')\|_{\cH^s}~dt' \\ & \ \  + \frac{mT^2}{2}\|V[\Psi^{(0)}]\Psi^{(0)}\|_{L^\infty_T\cH^s}, 
\end{align*}
where we used Minkowski inequality in the first inequality and H\"older's inequality in the second. 
We also used $0\leq \sqrt{-\Delta+m^2}-m\leq m$.

It follows from the fact that the nonlinearity is locally Lipschitz (Lemma \ref{lm:Lipschitz}) and (\ref{eq:NormBd1}) that 
\begin{align*}
& \|V[\Psi(t')]\Psi(t') - V[\Psi^{(0)}(t')]\Psi^{(0)}(t')\|_{\cH^s} \lesssim \rho^2 \|\Psi(t')-\Psi^{(0)}(t')\|_{\cH^s},\\
& \|V[\Psi^{(0)}]\Psi^{(0)}\|_{L^\infty_T\cH^s} \lesssim \rho^3.
\end{align*}
Hence
$$\|\Psi(t)-\Psi^{(0)}(t)\|_{\cH^s}\lesssim m\rho T + m\rho^3 T + \rho^2 \int_0^t \|\Psi(t')-\Psi^{(0)}(t')\|_{\cH^s} ~dt' .$$
By Gronwall's lemma, $\Psi\rightarrow \Psi^{(0)}$ strongly in $L^\infty_T(\cH^s)$ as $m\rightarrow 0.$ 
\end{proof}

\section{Asymptotic behaviour of solutions as mass tends to infinity}\label{sec:Nonrelativistic}

In this section, we discuss the asymptotics of the solution as the mass $m$ tends to infinity. 

\begin{proof}[Proof of Theorem \ref{th:InfiniteMass}]
Recall that from the proof of local well-posedness in Section \ref{sec:local}, there exists $T>0$ independent of $m$ such that $\|\Psi\|_{L^\infty_{T}\cH^s} \le C\|\Psi(0)\|_{\cH^s},$ where $C$ is independent of $m.$ 
Similarly, one can show that there exists $T'>0$ independent of $m$ such that $\|\Gamma\|_{L^\infty_{T'}\cH^s} \le C\|\Psi(0)\|_{\cH^s},$ where $C$ is independent of $m.$ Let $\tau = \min(T,T').$ 
Let $\tGamma = \{\tilde{\gamma}_k\}_{k\in{\mathbb N}}$ satisfy the system of equations 
\begin{equation*}
\begin{cases}
i\partial_t \tGamma=V[\tGamma]\tGamma,\\
V[\tGamma]=w_\gamma \star n[\tGamma],\ \  n[\tGamma]=\sum_{k=1}^{\infty}\lambda_{k}|\tilde{\gamma}_{k}|^{2},
\end{cases}
\end{equation*} 
with initial condition $\tGamma(0)=\Psi(0).$
Alternatively, $\tGamma$ satisfies the integral equation
$$\tGamma(t) = \Psi(0) -i\int_0^t V[\tGamma(t')]\tGamma(t')dt'.$$
Uniqueness of the solution follows from the fact that the nonlinearity is locally Lipschitz (Lemma \ref{lm:Lipschitz}).
We are going to compare $\Psi$ to $\tGamma,$ and then $\tGamma$ to $\Gamma.$

\begin{align*}
  \|\Psi(t)-\tGamma(t)\|_{\cH^s} &\le \|\left(U^{(m)}(t) -1 \right)\Psi(0)\|_{\cH^s} +
  \\&+ \int_0^t \|\left( U^{(m)}(t-t')-1 \right)V[\tGamma(t')]\tGamma(t')\|_{\cH^s} dt' + \\ &+ \int_0^t \|V[\Psi(t')]\Psi(t') - V[\tGamma(t')]\tGamma(t')\|_{\cH^s}dt'.
\end{align*}

To estimate the first term on the right-hand-side, we apply the Fourier transform and use Parseval's Theorem,
\begin{align*}
&\|\left(U^{(m)}(t) -1 \right)\Psi(0)\|^2_{\cH^s} 
\\&= \sum_{l\ge 1} \lambda_l \int_{\bbR^n} |e^{-it (\sqrt{m^2+ |k|^2} -m)} -1|^2 (1+|k|^2)^{s}|\widehat{\psi_l}(0,k)|^2 dk\\
&\le \sum_{l\ge 1} \lambda_l \{\int_{|k|\le m^{\frac{1}{4}}} |e^{-it (\sqrt{m^2+ |k|^2} -m)} -1|^2 (1+|k|)^{2s}|\widehat{\psi_l}(0,k)|^2 dk + 
\\ &+ \int_{|k|> m^{\frac{1}{4}}} |e^{-it (\sqrt{m^2+ |k|^2} -m)} -1|^2 (1+|k|)^{2s}|\widehat{\psi_l}(0,k)|^2 dk \}\\
&\le  \sum_{l\ge 1} \lambda_l \{\int_{|k|\le m^{\frac{1}{4}}} \frac{t^2|k|^4}{(\sqrt{m^2+|k|^2}+m)^2} (1+|k|)^{2s}|\widehat{\psi_l}(0,k)|^2 dk + \\ &+ 4\int_{|k|> m^{\frac{1}{4}}} (1+|k|)^{2s}|\widehat{\psi_l}(0,k)|^2 dk \}\\
&\le \frac{\tau^2}{4m}\|\Psi(0)\|_{\cH^s}^2 + 4\sum_{l\ge 1}\int_{|k|> m^{\frac{1}{4}}} (1+|k|)^{2s}|\widehat{\psi_l}(0,k)|^2 dk \\
&\rightarrow 0 \ \ {\mathrm as} \ \ m\rightarrow \infty.
\end{align*}

Since $V[\tGamma]\tGamma\in \cH^s,$ it follows from the Dominated Convergence Theorem that 
$$\lim_{m\rightarrow \infty} \int_0^t \|\left( U^{(m)}(t-t')-1 \right)V[\tGamma(t')]\tGamma(t')\|_{\cH^s} dt' =0.$$

To estimate the third term, let $\rho >0$ be a constant such that 
$$\sup_{m\ge 1} (\|\Psi\|_{L^\infty_\tau\cH^s} + \|\Gamma\|_{L^\infty_\tau\cH^s}) + \|\tGamma\|_{L^\infty_\tau\cH^s}\le \rho.$$
It follows from the fact that the nonlinearity is locally Lipschitz that 
$$ \|V[\Psi(t')]\Psi(t') - V[\tGamma(t')]\tGamma(t')\|_{\cH^s} \le C \rho^2 \|\Psi(t')-\tGamma(t')\|_{\cH^s},$$ where $C$ is a positive constant independent of $m.$

Therefore, 
$$\|\Psi(t)-\tGamma(t)\|_{\cH^s} \le f_m + C \rho^2\int_0^t   \|\Psi(t')-\tGamma(t')\|_{\cH^s} dt', $$
where $\lim_{m\rightarrow\infty} f_m =0$ and $C$ is independent of $m.$
It follows from Gronwall's lemma that $$\lim_{m\rightarrow\infty}\|\Psi-\tGamma\|_{L^\infty_\tau\cH^s} = 0.$$ 
Similarly, one can show that 
$$\|\Gamma(t)-\tGamma(t)\|_{\cH^s} \le g_m + C \rho^2\int_0^t   \|\Psi(t')-\tGamma(t')\|_{\cH^s} dt', $$
where $\lim_{m\rightarrow\infty} g_m =0$ and $C$ is independent of $m,$ and it follows that 
$$\lim_{m\rightarrow\infty}\|\Gamma-\tGamma\|_{L^\infty_\tau\cH^s} = 0.$$
Since 
$$\|\Psi-\Gamma\|_{L^\infty_\tau\cH^s}\le \|\Psi-\tGamma\|_{L^\infty_\tau\cH^s} + \|\Gamma-\tGamma\|_{L^\infty_\tau\cH^s},$$
it follows that 
$$\lim_{m\rightarrow\infty} \|\Psi-\Gamma\|_{L^\infty_\tau\cH^s} =0\,,$$
as desired.
\end{proof}

\section*{Acknowledgements}
WAS acknowledges the financial support of a Discovery grant from the Natural Sciences and Engineering Research Council of Canada. 
T.C. was supported by 
NSF grants DMS-1009448 and DMS-1151414 (CAREER). 


\appendix

\section{}
\label{sec:Preliminaries}

For the benefit of a general reader, we briefly recall in this section some useful results about fractional Leibniz rule and inequalities for fractional integral operator. In what follows, we denote $\cD = (-\Delta)^{1/2}.$ The proof can be found in the specified references.

The following result about the fractional Leibniz rule can be found in \cite{Kato95}.
\begin{lemma}\label{lm:fLeibniz}
$$\|\cD^s(uv)\|_{L^p} \lesssim \|\cD^s u \|_{L^{q_1}}\|v\|_{L^{r_1}} + \|u\|_{L^{q_2}}\|\cD^s v\|_{L^{r_2}},$$
where $\frac{1}{p} = \frac{1}{q_i}+\frac{1}{r_i}, \ \ i=1,2.$
\end{lemma}

The second result is about inequality involving fractional integral operators, which can be found, for example, in \cite{Stein93}. 
\begin{lemma}\label{lm:fIntegralOperator}
Let $I_\alpha$, for $0<\alpha<n$,
be the fractional integral operator 
$$I_\alpha (u) = \int_{\bbR^n} |x-y|^{\alpha-n} u(y)~dy.$$ Then 
$$\|I_\alpha (u)\|_{L^p} \lesssim \|u\|_{L^q}, \ \ \frac{1}{p} = \frac{1}{q}-\frac{\alpha}{n}.$$
\end{lemma}

We also recall the following useful Hardy-type inequality.

\begin{lemma}\label{lm:fSmoothing}
Let $0<\gamma<n$. Then,
$$\sup_{x\in {\mathbb R}^n} |\int_{{\mathbb R}^n} \frac{1}{|x-y|^\gamma} |u(y)|^2dy| \lesssim \|u\|_{\dot{H}^{\frac{\gamma}{2}}}^2 \,.$$
\end{lemma}


\end{document}